\documentclass[onecolumn]{IEEEtran}
\usepackage[utf8]{inputenc} 
\usepackage[T1]{fontenc} 
\usepackage{amsthm}
\usepackage[cmex10]{amsmath}
\usepackage{times,amssymb,amsfonts,float,nicefrac,color,mathrsfs,caption} 
\usepackage{enumerate,multirow,caption,tikz,graphicx,enumitem,soul}
\usepackage[mathscr]{eucal}
\usepackage{sidecap}
\usepackage{algpseudocode}
\usepackage{verbatim}
\usepackage{epstopdf}
\usepackage{textcomp}
\usetikzlibrary{shapes,arrows}
\usepackage[noadjust]{cite}
\usepackage{booktabs}
\usepackage{stmaryrd}
\allowdisplaybreaks

\usepackage[font={it}]{caption}
\usepackage[margin=0.05in]{subcaption}
\usepackage[bottom]{footmisc}
\usepackage[normalem]{ulem}

\usepackage[ruled,vlined,linesnumbered]{algorithm2e}

\usepackage{hyperref} 

\usepackage{balance}

\newtheoremstyle{mytheoremstyle}{3pt}{3pt}{\itshape}{}{\bf}{.}{.3em}{} 
\theoremstyle{mytheoremstyle}
\newtheorem{theorem}{Theorem}

\newcommand\nc\newcommand
\nc\bfa{{\boldsymbol a}}\nc\bfA{{\boldsymbol A}}\nc\cA{{\mathscr A}}
\nc\bfb{{\boldsymbol b}}\nc\bfB{{\boldsymbol B}}\nc\cB{{\mathscr B}}
\nc\bfc{{\boldsymbol c}}\nc\bfC{{\boldsymbol C}}\nc\cC{{\mathscr C}}
\nc\bfd{{\boldsymbol d}}\nc\bfD{{\boldsymbol D}}\nc\cD{{\mathscr D}}
\nc\bfe{{\boldsymbol e}}\nc\bfE{{\boldsymbol E}}\nc\cE{{\mathscr E}}
\nc\bff{{\boldsymbol f}}\nc\bfF{{\boldsymbol F}}\nc\cF{{\mathscr F}}
\nc\bfg{{\boldsymbol g}}\nc\bfG{{\boldsymbol G}}\nc\cG{{\mathscr G}}
\nc\bfh{{\boldsymbol h}}\nc\bfH{{\boldsymbol H}}\nc\cH{{\mathscr H}}
\nc\bfi{{\boldsymbol i}}\nc\bfI{{\boldsymbol I}}\nc\cI{{\mathcal I}}
\nc\bfj{{\boldsymbol j}}\nc\bfJ{{\boldsymbol J}}\nc\cJ{{\mathscr J}}
\nc\bfk{{\boldsymbol k}}\nc\bfK{{\boldsymbol K}}\nc\cK{{\mathscr K}}
\nc\bfl{{\boldsymbol l}}\nc\bfL{{\boldsymbol L}}\nc\cL{{\mathscr L}}
\nc\bfm{{\boldsymbol m}}\nc\bfM{{\boldsymbol M}}\nc{\cM}{{\mathscr M}}
\nc\bfn{{\boldsymbol n}}\nc\bfN{{\boldsymbol N}}\nc\cN{{\mathscr N}}
\nc\bfo{{\boldsymbol o}}\nc\bfO{{\boldsymbol O}}\nc\cO{{\mathscr O}}
\nc\bfp{{\boldsymbol p}}\nc\bfP{{\boldsymbol P}}\nc\cP{{\mathscr P}}\nc\eP{{\EuScriptP}}\nc\fP{{\mathfrak P}}
\nc\bfq{{\boldsymbol q}}\nc\bfQ{{\boldsymbol Q}}\nc\cQ{{\mathscr Q}}
\nc\bfr{{\boldsymbol r}}\nc\bfR{{\boldsymbol R}}\nc\cR{{\mathscr R}}
\nc\bfs{{\boldsymbol s}}\nc\bfS{{\boldsymbol S}}\nc\cS{{\mathscr S}}
\nc\bft{{\boldsymbol t}}\nc\bfT{{\boldsymbol T}}\nc\cT{{\mathscr T}}
\nc\bfu{{\boldsymbol u}}\nc\bfU{{\boldsymbol U}}\nc\cU{{\mathscr U}}
\nc\bfv{{\boldsymbol v}}\nc\bfV{{\boldsymbol V}}\nc\cV{{\mathscr V}}
\nc\bfw{{\boldsymbol w}}\nc\bfW{{\boldsymbol W}}\nc\cW{{\mathscr W}}
\nc\bfx{{\boldsymbol x}}\nc\bfX{{\boldsymbol X}}\nc\cX{{\mathscr X}}
\nc\bfy{{\boldsymbol y}}\nc\bfY{{\boldsymbol Y}}\nc\cY{{\mathscr Y}}
\nc\bfz{{\boldsymbol z}}\nc\bfZ{{\boldsymbol Z}}\nc\cZ{{\mathscr Z}}

\newtheorem{lemma}[theorem]{Lemma}
\newtheorem{claim}[theorem]{Claim}

\newtheorem{proposition}[theorem]{Proposition}

\newtheorem{definition}{Definition}

\newtheorem{example}{Example}[section]

\theoremstyle{remark}
\newtheorem{remark}{Remark}[section]

\newcommand{\etal}{{\em et al.\ }}

\DeclareMathOperator{\rank}{rank}

\newcommand{\ff}{{\mathbb F}}
\DeclareMathOperator{\wt}{wt}
\DeclareMathOperator{\rw}{rw}
\DeclareMathOperator{\arw}{arw}

\newcommand{\df}{d_{\mathrm{free}}}
\newcommand{\nz}[1]{\llbracket{#1}\rrbracket}
\newcommand{\pz}[1]{[{#1}]}

\begin{document}
	
	\title{Trellis codes with a good distance profile constructed from expander graphs}
	
	\author{
		\IEEEauthorblockN{Yubin Zhu} \hspace*{1in}
		\and
		\IEEEauthorblockN{Zitan Chen}
	}
	\maketitle	
	
	{\renewcommand{\thefootnote}{}\footnotetext{
			
			\vspace{-.2in}
			
			\noindent\rule{1.5in}{.4pt}

			{	
				Y.~Zhu is with School of Science and Engineering, The Chinese University of Hong Kong (Shenzhen), Shenzhen, Guangdong, 518172,  P.\,R.\,China (Email: yubinzhu@link.cuhk.edu.cn).			
				
				Z.~Chen is with School of Science and Engineering, Future Networks of Intelligence Institute, The Chinese University of Hong Kong (Shenzhen), Shenzhen, Guangdong, 518172, P.\,R.\,China (Email: chenztan@cuhk.edu.cn).
			}
		}
	}
	\renewcommand{\thefootnote}{\arabic{footnote}}
	\setcounter{footnote}{0}

	\begin{abstract} 
		We derive Singleton-type bounds on the free distance and column distances of trellis codes. Our results show that, at a given time instant, the maximum attainable column distance of trellis codes can exceed that of convolutional codes. Moreover, using expander graphs, we construct trellis codes over constant-size alphabets that achieve a rate-distance trade-off arbitrarily close to that of convolutional codes with a maximum distance profile. By comparison, all known constructions of convolutional codes with a maximum distance profile require working over alphabets whose size grows at least exponentially with the number of output symbols per time instant.
	\end{abstract}
	
	
	\section{Introduction}
	
	To encode a possibly infinite stream of data symbols using a length-$n$ block code, one may first divide the data steam into data blocks of length $k<n$, and then encode each data block into a code block of length $n$. In this block-coding setting, each code block relies only on a single data block. Time-dependent coding schemes generalize this approach by allowing each code block to causally depend on multiple preceding data blocks. Codes arising from such time-dependent coding schemes are called \emph{trellis codes}, in contrast to block codes. The subclass of trellis codes that are linear, analogous to block linear codes in the block-coding setting, is known as \emph{convolutional codes}. Below we review some basic concepts and properties of convolutional codes.
	
	\subsection{Convolutional codes}
	Let $m\geq 0$ be an integer and let $G_i,0\leq i\leq m$ be $k\times n$ matrices over the finite field $\ff_q$ of order $q$. A convolutional code over $\ff_q$ can be defined by the following semi-infinite matrix 
	\begin{equation*}
		G=\begin{pmatrix} G_0&G_1&\dots&G_m\\&G_0&G_1&\dots&G_m\\&&\ddots&\ddots&&\ddots\end{pmatrix},
	\end{equation*} where the empty blocks denote zero matrices. Given a semi-infinite sequence of length-$k$ vectors $x_t\in\ff_q^k,t\geq 0$, the sequence $(c_0,c_1,\ldots):=(x_0,x_1,\ldots)G$ is a codeword of the convolutional code generated by $G$, where $c_t=x_tG_0+x_{t-1}G_1+\cdots+x_{t-m}G_m\in\ff_q^n$ is the code block at time $t\geq 0$. 
	If the matrices $G_0,\ldots,G_m$ depend on the time $t$, then more precisely, one has $c_t=x_tG_0(t)+x_{t-1}G_1(t)+\cdots+x_{t-m}G_m(t)$ and the code is said to be \emph{time-varying}; otherwise it is called \emph{time-invariant}. In the sequel, we focus on time-invariant convolutional codes and refer to them simply as convolutional codes.
	
	In view of the encoding of $(x_0,x_1,\ldots)$ with $G$, convolutional codes can be thought of as nonblock linear codes over $\ff_q$. But it is often convenient to formulate convolutional codes in the language of polynomial rings.
	In particular, the sequence $(x_0,x_1,\ldots)$ can be associated with the polynomial $x(D)=\sum_{i=0}^{\infty}x_iD^i$. Similarly, the matrix $G$ can be identified with the polynomial $G(D)=\sum_{i=0}^{m}G_iD^i$ and the sequence $(c_0,c_1,\ldots)$ with the polynomial $c(D)=\sum_{i=0}^{\infty}c_iD^i$. It then follows that $c(D)=x(D)G(D)$. Let $\ff_q[D]$ be the ring of polynomials with indeterminate $D$ and coefficients in $\ff_q$.  Noticing $x_i\in\ff_q^k$, the polynomial $x(D)$ can also be viewed as a $k$-tuple with entries in $\ff_q[D]$. In a similar vein, $G(D),c(D)$ can be respectively viewed as a $k\times n$ matrix and an $n$-tuple with entries in $\ff_q[D]$. Formally, we have the following definition for convolutional codes from a module-theoretic perspective.
	
	\begin{definition}
		An $(n,k)$ convolutional code over $\ff_q$ is an $\ff_q[D]$-submodule $\cC\subset \ff_q[D]^n$ of rank $k$. A $k\times n$ polynomial matrix $G(D)$  over $\ff_q[D]$ is called a generator matrix of the code $\cC$ if 
		$
		\cC= \{x(D)G(D)\mid x(D)\in \ff_q[D]^k\}.
		$
	\end{definition} 
	The 
	$i$th \emph{row degree} of a generator matrix $G(D)=(g_{i,j}(D))$ of an $(n,k)$ convolutional code $\cC$ is defined to be 
	$
	\nu_i(G)=\max_{1\leq j\leq n} \deg g_{i,j}(D)
	$ for $i=1,\ldots,k$.
	Furthermore, the \emph{memory} of $G(D)$ is defined as 
	$
	\max_{1\leq i\leq k}\nu_i(G).
	$
	The \emph{overall constraint length} of $G(D)$ is defined to be 
	$
	\nu(G)=\sum_{i=1}^{k}\nu_i(G).
	$
	The matrix $G(D)$ is called \emph{reduced} (or \emph{minimal}) if its overall constraint length is the smallest over all generator matrices of the code $\cC$.
	This minimum overall constraint length is called the \emph{degree} of the code $\cC$, denoted by $\delta(\cC)$. Namely, we have
	$
	\delta(\cC) = \min\{\nu(G) \mid G(D)\text{ is a generator matrix of }\cC\}.
	$
	
	In direct analogy with block linear codes, one defines the \emph{free distance} of an $(n,k)$ convolutional code $\cC$ as the minimum Hamming weight of the nonzero codewords in $\cC$. More precisely, viewing a codeword $c(D)\in\cC$ as the sequence $(c_0,c_1,\ldots)$ over $\ff_q^n$, one defines the Hamming weight of $c(D)$ to be 
	$
	\wt(c(D)) = \wt(c_0,c_1,\ldots) = \sum_{t=0}^{\infty}\wt(c_t),
	$ where $\wt(c_t)$ denotes the Hamming weight of the vector $c_t\in \ff_q^n$. The {free distance} of $\cC$ is then defined to be
	\begin{align*}
		\df(\cC)=\min_{{c(D)\in\cC, c(D)\neq 0}}\wt(c(D)).
	\end{align*} 
	In addition to the free distance, one can also associate $\cC$ with another distance measure, termed \emph{column distance}. Roughly speaking, the column distance measures the minimum Hamming weight of truncations $(c_0,c_1,\ldots,c_j)$ of codewords $(c_0,c_1,\ldots)\in\cC$ with $c_0\neq 0$. 
	More precisely, assuming $G_0$ has full rank, the $j$th column distance of $\cC$ is defined as 
	\begin{align*}
		d_j^c(\cC)=\min_{{c(D)\in\cC,c_0\neq 0}}\wt(c_0,\ldots,c_j) .
	\end{align*}
	In \cite{rosenthal1999maximum}, Rosenthal and Smarandache showed that the free distance of an $(n,k)$ convolutional code $\cC$ satisfies the following upper bound that generalizes the Singleton bound for block codes:
	\begin{align}
		\df(\cC)\leq (n-k)\left( \left\lfloor\frac{\delta(\cC)}{k}\right\rfloor+1\right)+\delta(\cC)+1.\label{eq:sb}
	\end{align}
	Convolutional codes whose free distance attains \eqref{eq:sb} with equality are called maximum distance separable (MDS) convolutional codes.
	As for the column distances, it is shown in \cite{gluesing2006strongly} that for all $j\geq 0$,
	\begin{align}
		d_j^c(\cC)\leq (n-k)(j+1)+1.\label{eq:sb-c}
	\end{align} 
	Moreover, if equality holds for some $j$, then equality is also attained for all the previous distances $d_i^c(\cC),i< j$ in their corresponding versions of the bound \eqref{eq:sb-c}. 
	Convolutional codes whose column distances attain the bound \eqref{eq:sb-c} for the largest possible number of time instants are said to have a maximum distance profile (MDP), also known as MDP convolutional codes.
	Since $d_j^c(\cC)$ is a non-decreasing function of $j$ and $\lim_{j\to\infty} d_j^c(\cC)= \df(\cC)$, the column distance coincides with the free distance eventually. Therefore, the maximum number of time instants for which \eqref{eq:sb-c} can be attained is $L(\cC):=\lfloor\frac{\delta(\cC)}{k}\rfloor+\lfloor\frac{\delta(\cC)}{n-k}\rfloor$, which we refer to as the \emph{maximum profile length}. Accordingly, MDP convolutional codes are precisely those whose $L(\cC)$-th column distance attains the bound \eqref{eq:sb-c}. 
	Furthermore, the earliest time instant at which the bounds \eqref{eq:sb} and \eqref{eq:sb-c} coincide is $J(\cC):=\lfloor\frac{\delta(\cC)}{k}\rfloor+\lceil\frac{\delta(\cC)}{n-k}\rceil$, and convolutional codes whose $J(\cC)$-th column distance equals the right-hand side of the bound \eqref{eq:sb} are called strongly-MDS convolutional codes.
	
	\subsection{Trellis codes}
	We next turn to a brief introduction to trellis codes.
	Informally, a trellis code is formed of codewords generated by a trellis: a time-indexed digraph with vertices at each time and labeled edges between successive time-indexed vertex sets. Each valid codeword corresponds to a path through the trellis, obtained by reading off the edge labels along the path. More precisely, trellis codes can be described in terms of labeled digraphs. Let $\Sigma_q$ be a finite alphabet of size $q$ and let $\Gamma=(V,E,\psi)$ be a labeled digraph where $V$ is a nonempty finite set of vertices (also called \emph{states}), $E$ is a finite set of edges, and $\psi\colon E\to \Sigma_q^n$ is an edge-labeling function. The labeled digraph $\Gamma$ is \emph{irreducible} if for every ordered pair of states $(v_1,v_2)\in V\times V$ there is a path from $v_1$ to $v_2$ in $\Gamma$. Moreover, $\Gamma$ is said to be \emph{lossless}, if for any two distinct paths in $\Gamma$ with the same
	initial state and the same terminal state, the edge-labeling function $\psi$ induces different sequences of edge labels. The labeled digraph $\Gamma$ is called $M$-\emph{regular} if each state in $V$ has exactly $M$ outgoing edges.
	
	\begin{definition}
		An $(n,M)$ trellis code over $\Sigma_q$ is a set $\cC\subset(\Sigma_q^n)^\mathbb{N}$ consisting of all infinite sequences of edge labels induced by paths that starts from a fixed initial state in an irreducible, lossless, and $M$-regular labeled digraph $\Gamma$ with edge labels in $\Sigma_q^n$. 
		The digraph $\Gamma$ is called a generator (or presenting) digraph of the code $\cC$.
	\end{definition} 
	
	The \emph{rate} of an $(n,M)$ trellis code $\cC$ over $\Sigma_q$ is defined to be $\log_q M/n$. We refer to $M$ as the \emph{size} of $\cC$, and write $|\cC|:=M$.\footnote{Similarly, for an $(n,k)$ convolutional code $\cC$ over $\ff_q$, we refer to $q^k$ as the size of $\cC$ and write $|\cC|=q^k$.} Note that a trellis code may admit multiple generator digraphs, and thus the number of states needed to describe the trellis code need not be unique. With the operational meaning of the degree of convolutional codes in mind, we define the degree of a trellis code $\cC$ over $\Sigma_q$ as $\delta(\cC)=\min\{\log_q |V|: \Gamma=(V,E,\psi) \text{  is a generator digraph of } \cC\}$.
	
	Similarly to convolutional codes, we can also define the free distance and column distance for trellis codes. Let $c,c'$ be two codewords of an $(n,M)$ trellis code $\cC$ over $\Sigma_q$, and write $c=(c_0,c_1,\ldots),c'=(c'_0,c'_1,\ldots)$. The Hamming distance between $c,c'$ is defined as 
	$
	d(c,c') = \sum_{t=0}^{\infty}d(c_t,c'_t),
	$ where $d(c_t,c'_t)$ denotes the Hamming distance between the two $n$-tuples $c_t,c'_t$ over $\Sigma_q$. The {free distance} of $\cC$ is
	\begin{align*}
		\df(\cC)=\min_{{c,c'\in\cC, c\neq c'}}d(c,c').
	\end{align*} 
	Moreover, assume $\cC$ admits a generator digraph $\Gamma=(V,E,\psi)$ such that, for every state
	in $V$, no two outgoing edges have the same label (i.e., $\Gamma$ is \emph{deterministic}). The $j$th column distance of $\cC$ is then defined as 
	\begin{align*}
		d_j^c(\cC)=\min_{{c,c'\in\cC,c_0\neq c'_0}}\sum_{t=0}^{j}d(c_t,c'_t).
	\end{align*}

	\subsection{Our results and related work}
	
	It is tempting to conjecture that analogues of the Singleton-type bounds \eqref{eq:sb} and \eqref{eq:sb-c} on the free distance and column distances of convolutional codes also hold for trellis codes. However, to the best of our knowledge, no such results have appeared in the literature. In this work, we establish Singleton-type bounds on the free distance and column distances for trellis codes. Our results show that the bound \eqref{eq:sb}, with $k$ replaced by $\log_q M$, remains valid for the free distance of $(n,M)$ trellis codes over $\Sigma_q$. However, an analogous replacement in \eqref{eq:sb-c} does not hold in general for the column distances of trellis codes. In particular, the $j$th column distance of an $(n,M=q^k)$ trellis code over $\Sigma_q$ can exceed the maximum $j$th column distance of $(n,k)$ convolutional codes over $\ff_q$.
	Furthermore, our second contribution in this work is a construction of trellis codes over constant-size alphabets whose rate-distance trade-off can be made arbitrarily close to that of an MDP convolutional code. In contrast, all known constructions of MDP convolutional codes require an alphabet whose size grows at least exponentially with the block length $n$. 
	
	We next give an overview of related work on MDP convolutional codes. Gluesing-Luerssen \etal \cite{gluesing2006strongly} gave the first explicit construction of MDP convolutional codes over large finite fields of large characteristic. Generalizing the approach in \cite{gluesing2006strongly}, Almeida \etal \cite{almeida2013new} constructed another class of MDP convolutional codes over large finite fields of arbitrary characteristic. However, the field sizes required by \cite{almeida2013new} are substantially larger than those in \cite{gluesing2006strongly}. A third general class of explicit MDP convolutional codes was presented in \cite{alfarano2020weighted}, which in certain parameter regimes requires smaller finite fields than \cite{gluesing2006strongly}. 
	The aforementioned constructions apply for arbitrary maximum profile length. For the special case of unit maximum profile length, MDP convolutional codes over smaller finite fields were given in \cite{chen2022convolutional}, \cite{luo2023construction}, and \cite{chen2023lower}, among which the codes in \cite{chen2023lower} achieve the smallest field size currently known. However, as noted above, all of the known explicit constructions require a finite field of size at least exponential in $n$.
	Over the binary field, the recent work \cite{abreu2023binary} constructed convolutional codes with largest possible column distances based on simplex block codes.
	With the goal of reducing the encoding complexity, the authors of \cite{dang2025matrix} proposed a matrix completion approach for constructing unit-memory MDP convolutional codes.
	Regarding impossibility results on the field size required for the existence of MDP convolutional codes, a lower bound on the field size that holds for general code parameters was derived in \cite{chen2023lower}. Specifically, the author of \cite{chen2023lower} showed that a finite field with size growing polynomially in $n$ is necessary to construct MDP convolutional codes with constant rate and constant maximum profile length.
	
	We note that the trellis codes constructed in this work are defined over a constant-size alphabet, have a constant rate, and achieve a free distance whose growth rate is linear in $n$. Codes with such properties are often referred to as \emph{asymptotically good}. Asymptotically good convolutional codes have been studied in a series of earlier works. Specifically, asymptotic lower bounds on the free distance of convolutional codes, analogous to the Gilbert-Varshamov bound for block linear codes, were presented in \cite{costello1974free}, implying the existence of asymptotically good convolutional codes. In \cite{justesen1973new}, asymptotically good time-varying convolutional codes were explicitly constructed by applying the technique of constructing asymptotically good block linear codes in \cite{justesen1972class} on a class of convolutional codes generated by 
	cyclic Reed-Solomon codes. Building on the transitive algebraic geometry codes in \cite{stichtenoth2006transitive}, asymptotically good (time-invariant) convolutional codes are obtained in \cite{la2019asymptotically}. 
	The codes in \cite{justesen1973new} and \cite{la2019asymptotically} are constructed over a non-binary field. Another line of research \cite{sridharan2007distance,truhachev2010distance,mitchell2012minimum} is devoted to the study of asymptotically good low-density parity-check (LDPC) convolutional codes over the binary field. In particular, various ensembles of LDPC convolutional codes were analyzed in \cite{sridharan2007distance,truhachev2010distance,mitchell2012minimum} and asymptotic lower bounds on the free distance for the ensembles were showed to grow linearly with the decoding constraint length, thereby suggesting the existence of asymptotically good binary convolutional codes in the ensembles. However, a brute-force search in the ensembles may be required to identify the asymptotically good codes.
	
	In this work, our trellis codes are constructed using {expander graphs}, in particular {spectral expanders}. Roughly speaking, an expander graph is a graph with good expansion properties: every not-too-large set of vertices has many neighbors. Explicit constructions of spectral expanders were given in \cite{Mar73,lubotzky1988ramanujan} and have been used in numerous works to construct explicit asymptotically good block codes with efficient encoding/decoding, including \cite{alon1992construction,sipser1996expander,guruswami2005linear,barg2005concatenated,roth2006improved}, among others. Building on ideas from \cite{roth2006improved}, we show that one can combine a convolutional code with a block code on spectral expanders to obtain a trellis code over an alphabet of constant size, whose rate-distance trade-off almost matches that of the underlying convolutional code.

	\paragraph*{Paper outline}
	The remainder of the paper is organized as follows. Section~\ref{sec:bounds} presents Singleton-type bounds on the free distance and column distances for trellis codes. Section~\ref{sec:construction} is devoted the construction of trellis codes based on expander graphs and the analysis of their parameters. In Section~\ref{sec:re}, we point out a few future directions.

	\section{Distance bounds for trellis codes}\label{sec:bounds}
	
	In this section, we present distance bounds for trellis codes. We begin with a Singleton-type bound for the free distance of trellis codes.
	
	\begin{theorem}
		The free distance of an $(n,M)$ trellis code $\cC$ over $\Sigma_q$ satisfies
		\begin{align}
			\df(\cC)\leq (n-\log_q M)\left( \left\lfloor\frac{\delta(\cC)}{\log_q M}\right\rfloor+1\right)+\delta(\cC)+1.\label{eq:tsb}
		\end{align}
	\end{theorem}
	
	\begin{proof}
		Let $\Gamma=(V,E,\psi)$ be a generator digraph of $\cC$ with initial state $v_0\in V$, and define $\ell$ to be the smallest positive integer such that $M^\ell > |V|$, i.e., $\ell:=\lfloor \log_M|V|\rfloor+1$.
		Consider all length-$\ell$ paths in $\Gamma$ starting at $v_0$. Since $\Gamma$ is $M$-regular, the number of such paths is $M^\ell$. Each of these paths terminates at some state in $V$. Therefore, by the Pigeonhole principle, there exists a state $v_\ell\in V$ that is the terminal state of at least $\lceil M^\ell/|V|\rceil$ of these paths. Next, select a subset of size $\lceil M^\ell/|V|\rceil$ of these length-$\ell$ paths that start at $v_0$ and end at $v_\ell$. The corresponding sequences of edge labels form a block code $\cB$ of length $\ell n$ and size $\lceil M^\ell/|V|\rceil \geq 2$ over $\Sigma_q$. By the Singleton bound for block codes, the minimum distance of $\cB$ satisfies
		\begin{align}
			d(\cB)&\leq \ell n - \log_q\left\lceil \frac{M^\ell}{|V|}\right\rceil+1\nonumber\\
			&\leq \ell n - \ell \log_q {M}+\log_q{|V|}+1\nonumber\\
			&= (n - \log_q{M})\left( \left\lfloor\frac{\log_q|V|}{\log_q M}\right\rfloor+1\right)+\log_q{|V|}+1,\nonumber
		\end{align} where the last step follows by substituting $\ell = \lfloor \log_M|V|\rfloor+1=\lfloor{\log_q|V|}/{\log_q M}\rfloor+1$. By definition of $\df(\cC)$, we have 
		\begin{align*}
			\df(\cC)\leq d(\cB)\leq (n - \log_q{M})\left( \left\lfloor\frac{\log_q|V|}{\log_q M}\right\rfloor+1\right)+\log_q{|V|}+1.
		\end{align*} Furthermore, the above inequality holds for any generator digraph of $\cC$. In particular, applying it to a generator digraph with $\log_q|V|=\delta(\cC)$ yields the bound \eqref{eq:tsb}.
	\end{proof}
	
	We next give an example to show that, for $j\geq 1$, the $j$th column distance of an $(n,M)$ trellis code over $\Sigma_q$ can exceed the bound \eqref{eq:sb-c}, with $k$ replaced by $\log_q M$. For an $(n,M)$ trellis code $\cC$ over $\Sigma_q$, define $\cC_j=\{(c_0,\ldots,c_j)\mid (c_0,c_1,\ldots)\in\cC\}$ to be the block code obtained by truncating $\cC$ to its first $(j+1)n$ symbols. 
	
	\begin{example}\label{ex:large-cd}
		Fix $j\geq 1$.
		Assume $M^j\mid q$ and $(q/M^j)^n=M$. Clearly, such integers $q,M$ exist and satisfy $M\leq q^n$. 
		Let $\{A_{j,t}\}_{1\leq t\leq M^j}$ be a partition of $\Sigma_q$ with $|A_{j,t}|=q/M^j$. For $1\leq t\leq M^j$, define $\cA_{j,t}=A_{j,t}^n$. 
		For $0\leq i\leq j-1$ and $1\leq t\leq M^{i}$, define $\cA_{i,t}$ by choosing exactly one element from each of $\cA_{i+1,(t-1)M+1},\cA_{i+1,(t-1)M+2},\ldots,\cA_{i+1,tM}$. Then $|\cA_{i,t}|=(q/M^j)^n=M$ for all $0\leq i\leq j$ and $1\leq t\leq M^i$. Moreover, for each fixed $i$, the effective alphabets of $\cA_{i,1},\ldots,\cA_{i,M^i}$ are disjoint because  $\{A_{j,t}\}_{1\leq t\leq M^j}$ is a partition of $\Sigma_q$.
		
		For $0\leq i\leq j-1$ and $a\in\bigcup_{s=1}^{M^i}\cA_{i,s}$, let $\Lambda_i(a)$ be the unique set among $\{\cA_{i+1,t}\}_{1\leq t\leq M^{i+1}}$ that contains $a$.
		Define a block code $\cA\subset \Sigma_q^{(j+1)n}$ by 
		\begin{align*}
			\cA=\{(a_0,\ldots,a_j)\mid a_0\in \cA_{0,1}; a_i\in \Lambda_{i-1}(a_{i-1}) \text{ for $1\leq i\leq j$}\}.
		\end{align*} Since $|\cA_{0,1}|=M$ and $\Lambda_{i-1}(a_{i-1})$ has size $M$ for $1\leq i\leq j$, we have $|\cA|=M^{j+1}$.
		Let $\cC$ be an $(n,M)$ trellis code over $\Sigma_q$ with $\cC_j=\cA$. By construction, for any $(a_0,\ldots,a_j),(a'_0,\ldots,a'_j)\in \cC_j$ with $a_0\neq a'_0$, the $(j+1)n$ symbols of $(a_0,\ldots,a_j)$ and $(a'_0,\ldots,a'_j)$ are all distinct. 
		Indeed, for every $0\leq i\leq j$ the $n$-tuples $a_i$ and $a'_i$ lie over disjoint alphabets.
		Thus, $d_j^c(\cC)=(j+1)n$. At the same time, for $n>1$, the bound \eqref{eq:sb-c} with $k$ replaced by $\log_q M$ is $(j+1)n-(j+1)\log_q M+1=(j+1)n-\frac{(j+1)n}{jn+1}+1< (j+1)n$.
	\end{example}
	
	Although the bound \eqref{eq:sb-c} does not hold in general for trellis codes, we show in Theorem~\ref{thm:sb-c-trellis-1} below that if the $j$th column distance is known to be sufficiently small, a bound similar to \eqref{eq:sb-c} still applies.
	
	\begin{theorem}\label{thm:sb-c-trellis-1}
		For each $j\geq 0$, if the $j$th column distance of an $(n,M)$ trellis code $\cC$ over $\Sigma_q$ satisfies $d_j^c(\cC)\leq n$ then it further holds that
		\begin{align}
			d_j^c(\cC)\leq (j+1)(n-\log_q M)+1.\label{eq:tsbc}
		\end{align}
		Moreover, if equality holds for some $j>0$, then $d_i^c(\cC)$ attains the corresponding version of the bound \eqref{eq:tsbc} for $0\leq i<j$.
	\end{theorem}
	\begin{proof}
		Since $\cC$ admits a deterministic generator digraph, the block code $\cC_j$ obtained by truncating $\cC$ to its first $(j+1)n$ symbols has size $|\cC_j|=M^{j+1}$. 
		For each $a\in\Sigma_q^{jn}$, let $\cC_j(a)=\{c_0\mid (c_0,\ldots,c_j)\in\cC_j,(c_1,\ldots,c_j)=a\}$ be the set of length-$n$ prefixes of the codewords in $\cC_j$ with the length-$(jn)$ suffix equal to $a$. 
		If $|\cC_j(a)|\geq 2$, by the definition of the $j$th column distance, the minimum distance $d(\cC_j(a))$ of $\cC_j(a)$ satisfies $n\geq d(\cC_j(a))\geq d_j^c(\cC)$. Moreover, by the Singleton bound for block codes, $|\cC_j(a)|\leq q^{n-d_j^c(\cC)+1}$.
		Noticing that $\{\cC_j(a)\mid a\in\Sigma_q^{jn}\}$ induces a partition on $\cC_j$, we further obtain
		\begin{align*}
			M^{j+1}=|\cC_j|= \sum_{a\in\Sigma_q^{jn}}|\cC_j(a)|\leq q^{jn}\cdot q^{n-d_j^c(\cC)+1}=q^{(j+1)n-d_j^c(\cC)+1}.
		\end{align*}
		Hence, $d_j^c(\cC)\leq (j+1)(n-\log_q M)+1$. 
		Now suppose that equality in \eqref{eq:tsbc} holds, i.e., 
		$|\cC_j|=M^{j+1}=q^{(j+1)n-d_j^c(\cC)+1}$ for some $j>0$, but $|\cC_{j-1}|=M^j<q^{jn-d_{j-1}^c(\cC)+1}$. Then
		\begin{align*}
			\frac{|\cC_j|}{|\cC_{j-1}|}=M> q^{n}.
		\end{align*} However, since $\cC$ admits a deterministic generator digraph, the $M$ outgoing edges of any state in $V$ must carry distinct labels in $\Sigma_q^n$. Therefore, $M\leq q^n$, contradicting $M> q^{n}$.
	\end{proof}
	
	More generally, by a straightforward application of the Singleton bound for block codes, we have the following bound on the column distances of trellis codes.
	\begin{theorem}
		For each $j\geq 0$, the $j$th column distance of an $(n,M)$ trellis code $\cC$ over $\Sigma_q$ satisfies
		\begin{align}
			d_j^c(\cC)\leq (j+1)n-\log_qM+1.\label{eq:tsbc-2}
		\end{align}
		Moreover, if equality holds for some $j>0$, then $d_i^c(\cC)$ attains the corresponding version of the bound \eqref{eq:tsbc-2} for $0\leq i<j$.
	\end{theorem}
	\begin{proof}
		Since $\cC$ admits a deterministic generator digraph, we may take $M$ codewords from $\cC_j$ that correspond to the $M$ distinct paths of length $j+1$ starting from the initial state of the digraph. By definition of the $j$th column distance, the minimum Hamming distance between distinct pairs of these $M$ codewords is no smaller than the $j$th column distance of $\cC$. At the same time, these $M$ codewords form a block code length $(j+1)n$ and size $M$ over $\Sigma_q$, and therefore, the bound \eqref{eq:tsbc-2} follows. Now if $d_j^c(\cC)= (j+1)n-\log_qM+1$ but $d_{j-1}^c(\cC)<jn-\log_qM+1$ where $j>0$, then we have $d_j^c(\cC)-d_{j-1}^c(\cC)>n$, which contradicts the fact that the column distances can increase by at most $n$ from one time instant to the next.
	\end{proof}
	
	\begin{remark}
		We note that the bound \eqref{eq:tsbc-2} is in fact tight. Fix $j\geq 1$. For $(n,M)$ trellis codes $\cC$ over $\Sigma_q$ with $M^j\mid q$ and $(q/M^j)^n=M$, the bound \eqref{eq:tsbc-2} gives $d_j^c(\cC)\leq (j+1)n-\frac{n}{jn+1}+1$. As $d_j^c(\cC)$ is an integer, the bound \eqref{eq:tsbc-2} in effect gives $d_j^c(\cC)\leq (j+1)n$. At the same time, the $j$th column distance of the $(n,M)$ trellis code in Example~\ref{ex:large-cd} is equal to $(j+1)n$.
	\end{remark}

	\section{Trellis codes on expander graphs}\label{sec:construction}
	In this section, we present a construction of trellis codes based on bipartite spectral expanders. 
	We begin by reviewing some basic concepts and outlining the high-level ideas behind the construction.
	A graph is $\Delta$-\emph{regular} if every vertex in the graph has degree $\Delta$. Moreover, a graph on $n$ vertices is said to be a $(\Delta,\gamma)$ \emph{spectral expander} if the graph is $\Delta$-regular and the second largest eigenvalue (in absolute value) of its adjacency matrix is at most $\Delta\gamma\leq \Delta$. As smaller values of $\gamma$ implies better expansion properties, in many applications of spectral expanders, one requires $\gamma$ to be vanishing as $\Delta$ grows. Remarkably, this can be accomplished by explicit constructions, notably Ramanujan graphs \cite{lubotzky1988ramanujan}. For coding-theoretic applications, one often further requires that the Ramanujan graphs are bipartite. 
	
	Using bipartite Ramanujan graphs, Roth and Skachek \cite{roth2006improved} showed that one can construct asymptotically good block codes from two short block linear codes that preserves almost the same rate-distance trade-off as one of them, provided that the other has rate close to one. Generalizing this idea, our construction below is also defined via bipartite Ramanujan graphs, with the two constituent codes taken to be a convolutional code and a block linear code with rate close to one. For a convolutional code with generator matrix $G(D)=\sum_{j=0}^{m}G_jD^j$, one natural idea for combining it with a block code on a bipartite spectral expander is to combine, for each $j$, the code generated by $G_j$ with the block code. To preserve the rate-distance trade-off of the convolutional code, however, the key challenge is to maintain the correlations among code blocks across different time instants. We address this by using a sequence of isomorphic bipartite spectral expanders and enforcing a consistent total ordering on the edges of each expander. This consistency preserves the desired inter-time correlations, and in turn enables us to show the relative column distances of the resulting trellis codes constructed on the expanders are close to those of the underlying convolutional code.
	
	\subsection{Code construction}
	We first introduce the notation for the graphs underlying our construction.
	For an integer $a$, define $\pz{a}=\{1,2,\ldots,a\}$ and $\nz{a}=\{0,1,\ldots,a\}$.
	Let $\cG_0=(U_0,V_0,E_0)$ be a $\Delta$-regular balanced bipartite graphs on $2n$ vertices. 
	Assume an arbitrary but fixed total order on the set of edges $E_0$.
	Let $\cG_1=(U_1,V_1,E_1),\cG_2=(U_2,V_2,E_2),\ldots$ be a sequence of isomorphic copies of the graph $\cG_0$ such that the vertex sets $U_0\cup V_0,U_1\cup V_1,\ldots$ are disjoint.
	In other words, $\cG_0,\cG_1,\ldots$ refer to the same graph but the vertex labels are all distinct. 
	For $j\geq 0$, write $U_j=\{u_{j,1},\ldots,u_{j,n}\}$ and $V_j=\{v_{j,1},\ldots,v_{j,n}\}$. For ease of exposition, we further assume the graphs $\cG_1,\cG_2,\ldots$ are such that $(u_{j,s},v_{j,t})\in E_{j}$ if and only if $(u_{0,s},v_{0,t})\in E_{0}$, where $j\geq 1$ and $s,t\in\pz{n}$. Moreover, for each $j\geq 1$ we assume an ordering on $E_j$ that is \emph{consistent} with $E_0$ in the sense that for any two edges $(u_{j,s},v_{j,t}),(u_{j,s'},v_{j,t'})$ in $E_j$, $(u_{j,s},v_{j,t})$ precedes $(u_{j,s'},v_{j,t'})$ if and only if $(u_{0,s},v_{0,t})$ precedes $(u_{0,s'},v_{0,t'})$ in $E_0$, where $s,s',t,t'\in\pz{n}$.
	For a vertex $w\in U_j\cup V_j$ with $j\geq 0$, denote by $\cE(w)$ the set of edges incident with $w$. Since $\cG_j$ is $\Delta$-regular, we have $|\cE(w)|=\Delta$ for all $w\in U_j\cup V_j$. Moreover, for any subset $W\subset \bigcup_{j=0}^{m}(U_j\cup V_j)$, define $\cE(W)=\bigcup_{w\in W}\cE(w)$.
	
	We next define the notion of {relative weight} that will be useful later. Let $c\in \ff_q^n$ be a length-$n$ vector over $\ff_q$.
	The \emph{relative weight} of $c$ is defined to be $\rw(c):=\wt(c)/n$. For any nonempty subset $S\subset \ff_q^n$, the \emph{average relative weight} of $S$ is defined to be
	\begin{align*}
		\arw(S)=\frac{\sum_{c\in S}\rw(c)}{|S|}.
	\end{align*}
	Define $\arw(\emptyset)=0$.
	
	We are now ready to describe the construction.
	Let $G(D)=\sum_{j=0}^{m}G_jD^j$ be a generator matrix for a $(\Delta,k)$ convolutional code $\cC$ over $\ff_q$ where $G_0$ has full rank. Let $\cB_1$ be an $((m+1)\Delta, k)$ block linear code over $\ff_q$ defined by the generator matrix $(G_0,\ldots,G_m)$. 
	Let $\cB_2$ be a $(\Delta, r\Delta)$ block linear code with rate $r$ and relative minimum distance $\theta$.
	Define a block code of length $n(m+1)\Delta$ over $\ff_q$ whose coordinates are indexed by the set of edges $\bigcup_{j=0}^{m}E_j$:
	\begin{align}
		\cB &:=\Big\{c\in \ff_q^{n(m+1)\Delta} \mid  c_{\cE(v_{0,s},\ldots,v_{m,s})}\in\cB_1,c_{\cE(u_{j,s})}\in \cB_2,s\in\pz{n},j\in\nz{m}\Big\}\label{eq:def:B}\\
		&=\Big\{c\in \ff_q^{n(m+1)\Delta} \mid  c_{\cE(v_{0,s},\ldots,v_{m,s})}\in\cB_1,c_{\cE(u_{0,s},\ldots,u_{m,s})}\in \cB_2^{m+1},s\in\pz{n}\Big\},\nonumber
	\end{align} where $\cB_2^{m+1}$ is the $(m+1)$th Cartesian power of $\cB_2$.
	In other words, for any $c\in\cB$, the projection of $c$ onto the coordinates indexed by edges incident with $\{v_{j,s}\mid j\in\nz{m}\}$ is a codeword of $\cB_1$, and the projection onto the coordinates indexed by edges incident with $u_{j,s}$ is a codeword of $\cB_2$.
	It is easy to verify that $\cB$ is an $\ff_q$-linear code. In fact, we have $\cB=\cB^{(1)}\cap\cB^{(2)}$ where 
	\begin{align}
		\cB^{(1)}&:= \Big\{c\in \ff_q^{n(m+1)\Delta} \mid  c_{\cE(v_{0,s},\ldots,v_{m,s})}\in\cB_1,s\in\pz{n} \Big\},\label{eq:bc1}\\
		\cB^{(2)} &:= \Big\{c\in \ff_q^{n(m+1)\Delta} \mid c_{\cE(u_{0,s},\ldots,u_{m,s})}\in \cB_2^{m+1},s\in\pz{n}\Big\}.\label{eq:bc2}
	\end{align}
	
	Since we will define a convolutional code over $\ff_q$ of length $n\Delta$ and memory $m$ using a generator matrix of $\cB$, the following proposition will be helpful.
	\begin{proposition}\label{prop:B}
		Let $\tilde{k}$ denote the dimension of $\cB$.
		The code $\cB$ has a generator matrix in the form $(\tilde{G}_0,\ldots,\tilde{G}_m)$ where $\tilde{G}_j$ is a $\tilde{k}\times n\Delta$ matrix over $\ff_q$ with the columns being indexed (and ordered) by the edges in $E_j$ for $j\in\nz{m}$. Moreover, $\tilde{G}_0$ has full rank.
	\end{proposition}
	
	\begin{proof}
		It is clear that by definition $\cB$ is $\ff_q$-linear code of length $n(m+1)\Delta$, and thus it has a $\tilde{k}\times n(m+1)\Delta$ generator matrix over $\ff_q$ in the form of  $(\tilde{G}_0,\ldots,\tilde{G}_m)$. To see that $\tilde{G}_0$ has full rank, it suffices to show that the restriction of any nonzero codeword $c\in\cB\setminus\{0\}$ to the coordinates indexed by $E_0$ is nonzero, i.e., $c_{E_0}\neq 0$. Let $c\in\cB\setminus\{0\}$. If $c_{E_0}\neq 0$ there is nothing to prove. So let us further assume that $c_{E_0}=0$. It follows that $c_{\cE(v_{0,s})}=0$ for all $s\in\pz{n}$. Since $G_0$ has full rank, by the construction of $\cB$, it further follows that $c_{\cE(v_{j,s})}=0$ for $j\in\pz{m}$ and $s\in\pz{n}$, contradicting that $c\neq 0$.
	\end{proof}
	
	Define an $(n\Delta,\tilde{k})$ convolutional code $\tilde{\cC}$ over $\ff_q$ with generator matrix $\tilde{G}(D):=\sum_{j=0}^m\tilde{G}_jD^j$ where $\tilde{G}_0,\ldots,\tilde{G}_m$ are as in Proposition~\ref{prop:B}. We next describe the construction of a trellis code over $\ff_{q^{r\Delta}}$ from $\tilde{\cC}$.
	
	Let $c(D)=\sum_{j=0}^{\infty}c_j D^j\in \ff_{q}^{n\Delta}[D]$ be a codeword in $\tilde{\cC}$.
	Note that for each $j\geq 0$, we have $c_j=x_j\tilde{G}_0+x_{j-1}\tilde{G}_1+\cdots+x_{j-m}\tilde{G}_m\in\ff_q^{n\Delta}$ for some $x_j,x_{j-1},\ldots,x_{j-m}\in\ff_q^{\tilde{k}}$ with $x_{j-i}=0$ if $j-i<0$. 
	Since the columns of $\tilde{G}_0,\ldots,\tilde{G}_m$ are indexed by $E_0,\ldots,E_m$, respectively, and the orderings on the sets $E_j,j\geq 0$ are all consistent, we may simply view the columns of $\tilde{G}_i,i\in\nz{m}$ as indexed by $E_j$ for any $j\geq 0$. This viewpoint will be crucial for our subsequent discussion. In particular, it allows us to regard the coordinates of $c_j\in\ff_q^{n\Delta}$ as indexed by $E_j$, i.e., as lying on the edges of $\cG_j$ for each $j\geq 0$.
	Moreover, as a consequence of the consistent ordering of the edges, we have the following claims, whose proofs are given in the Appendix.
	
	\begin{claim}\label{cl:trellis-B2}
		Let $c(D)=\sum_{j=0}^{\infty}c_j D^j\in \ff_{q}^{n\Delta}[D]$ be a codeword in $\tilde{\cC}$. Then $(c_j)_{\cE(u_{j,s})}\in\cB_2$ for all $s\in\pz{n}$.
	\end{claim}

	\begin{claim}\label{cl:trellis-C}
		Let $c(D)=\sum_{j=0}^{\infty}c_j D^j\in \ff_{q}^{n\Delta}[D]$ be a codeword in $\tilde{\cC}$. Then $((c_0)_{\cE(v_{0,s})},\ldots,(c_j)_{\cE(v_{j,s})})$ is a codeword in $\cC$ truncated at time instant $j$ where $j\geq 0$ and $s\in\pz{n}$.
	\end{claim}
	
	Let $\phi\colon\cB_2\to\ff_{q^{r\Delta}}$ be an $\ff_q$-linear bijection. By Claim~\ref{cl:trellis-B2} we have $(c_j)_{\cE(u_{j,s})}\in\cB_2$ and so $\phi((c_j)_{\cE(u_{j,s})})\in \ff_{q^{r\Delta}}$ for $s\in\pz{n}$. Let $C_j=\phi^n(c_j):=(\phi((c_j)_{\cE(u_{j,1})}),\ldots,\phi((c_j)_{\cE(u_{j,n})})\in \ff_{q^{r\Delta}}^n$ be a length-$n$ vector over $\ff_{q^{r\Delta}}$. 
	
	Finally, define a trellis code $\tilde{\cC}_{\phi}\subset \ff_{q^{r\Delta}}[D]^n$ by
	\begin{align*}
		\tilde{\cC}_{\phi}=\Big\{\sum_{i=0}^{\infty}C_iD^i \mid C_i=\phi^n(c_i),\sum_{j=0}^{\infty}c_jD^j\in\tilde{\cC}\Big\}.
	\end{align*}
	Note that since $\phi$ is $\ff_q$-linear, $\tilde{\cC}_{\phi}$ is an $\ff_q$-linear trellis code over $\ff_{q^{r\Delta}}$. 
	
	\begin{remark}
		We note that the alphabet size $q^{r\Delta}$ of the trellis code $\tilde{\cC}_{\phi}$ is a constant independent of $n$. Indeed, $r$ and $\Delta$ do not depend on $n$, and $q$ depends only on the code parameters of $\cC$ and $\cB_2$, all of which are independent of $n$. 
	\end{remark}
	
	\subsection{Code parameters}
	In this subsection we discuss the column distances, size and degree of the trellis code $\tilde{\cC}_{\phi}$. Omitted proofs of the results in this subsection can be found in the Appendix.
	
	\begin{lemma}[The column distances of $\tilde{\cC}_{\phi}$]\label{le:cd}
		Let $\cG_0$ be a $(\Delta,\gamma)$ spectral expander. Then the $j$th column distance $d_j^c(\tilde{\cC}_{\phi})$ of the trellis code $\tilde{\cC}_{\phi}$ satisfies 
		\begin{align*}
			\frac{d_j^c(\tilde{\cC}_{\phi})}{(j+1)n}\geq 
			\frac{\min_{0\leq i\leq j}\Big\{\frac{d_i^c(\cC)}{(i+1)\Delta}\Big\}-
				\gamma\sqrt{\theta^{-1}}}{1-\gamma}.
		\end{align*}
	\end{lemma}
	
	\begin{proof}
		Let $c(D)=\sum_{i=0}^{\infty}c_i D^i$ be a codeword in $\tilde{\cC}$ {with $c_0\neq 0$}. For $j\geq 0$, define 
		\begin{align*}
			S_j &= \{s
			\mid (c_j)_{\cE(u_{j,s})}\neq 0,s\in\pz{n}\},\\
			T_j &= \{s
			\mid (c_j)_{\cE(v_{j,s})}\neq 0,s\in\pz{n}\},\\
			Y_j &= \{e\mid (c_j)_e\neq 0,e\in E_j\}.
		\end{align*} As discussed earlier, the coordinates of $c_j$ can be viewed as lying on the edges in $E_j$ for $j\geq 0$. So informally speaking, $S_j$ (resp., $T_j$) is the set of vertices in $U_j$ (resp., $V_j$), each of which sees a nonzero subword of $c_j$ on its incident edges, and $Y_j$ is the subset of $E_j$ consisting of edges, each of which carries a nonzero symbol of $c_j$.
		
		{Since $\tilde{\cC}_{\phi}$ is $\ff_q$-linear, the $j$th column distance of $\tilde{\cC}_{\phi}$ equals the minimum Hamming weight over $\ff_{q^{r\Delta}}$ among the codewords in $\tilde{\cC}_{\phi}$ whose first code block is nonzero.}
		Observe that $\wt(C_j)=\wt(\phi^n(c_j))=|S_j|$. Moreover, since $\phi$ is $\ff_q$-linear and $c_0\neq 0$, we have $C_0\neq0$. 
		It follows that the $j$th column distance of $\tilde{\cC}_{\phi}$ is 
		\begin{align*}
			{d_j^c(\tilde{\cC}_{\phi})}={\min_{c(D)\in\tilde{\cC}:c_0\neq 0}}{\sum_{i=0}^{j}|S_i|}.
		\end{align*}
		
		Next, we derive a lower bound on $|S_j|$ using the expansion properties of $\cG_0$ together with the distances properties of $\cB_1$ and $\cB_2$. By Claim~\ref{cl:trellis-B2}, we have $(c_j)_{\cE(u_{j,s})}\in\cB_2$ for all $s\in\pz{n}$. Since $d(\cB_2)\geq \theta \Delta$, we have $|Y_j|\geq |S_j|\cdot\theta\Delta$. 
		At the same time, we have $|Y_j|=|T_j|\cdot\arw(\bar{T}_j)\Delta$, where $\bar{T}_j:=\{(c_j)_{\cE(v_{j,s})}\mid s\in T_j\}$ is the set of the nonzero subwords of $c_j$ indexed by $T_j$. By a version of the expander mixing lemma, we have the following result that relates the size of $S_j$ with the average relative weight of $\bar{T}_j$.
		
		\begin{lemma}\label{le:rw}
			For all $j\geq 0$, we have
			\begin{align*}
				\frac{|S_j|}{n}\geq \frac{\arw(\bar{T}_j)-\gamma\sqrt{\theta^{-1}}}{1-\gamma}.
			\end{align*}
		\end{lemma}
		
		By Lemma~\ref{le:rw}, the $j$th relative column distance of $\tilde{\cC}_{\phi}$ satisfies 
		\begin{align*}
			\frac{d_j^c(\tilde{\cC}_{\phi})}{(j+1)n}
			&\geq {\min_{c(D)\in\tilde{\cC}:c_0\neq 0}}\frac{\sum_{i=0}^{j}\arw(\bar{T}_i)-(j+1)\gamma\sqrt{\theta^{-1}}}{(j+1)(1-\gamma)}.
		\end{align*}
		Therefore, it remains to show
		\begin{align*}
			\frac{\sum_{i=0}^{j}\arw(\bar{T}_i)}{j+1}\geq \min_{0\leq i\leq j}\Big\{\frac{d_i^c(\cC)}{(i+1)\Delta}\Big\}.
		\end{align*}
		In fact, the weights of the nonzero subwords of $\{c_i\mid i\in\nz{j}\}$ indexed by $\{T_i\mid i\in\nz{j}\}$ are correlated through $\cB_1$ (or rather, the convolutional code $\cC$), and a careful analysis of these correlations yields the following result.
		\begin{lemma}\label{le:partition}
			For $j\geq 0$, we have
			\begin{align}
				{\sum_{i=0}^{j}\arw(\bar{T}_i)\Delta}
				&\geq (j+1)\sum_{i=0}^j \frac{\lambda_i d_i^c({\cC})}{i+1},\label{eq:p}
			\end{align} 	
			where $\lambda_i\geq 0$ for all $i\in\nz{j}$ and $\sum_{i=0}^j\lambda_i=1$.
		\end{lemma}
		
		As a consequence of Lemma~\ref{le:partition}, the summation in the right-hand side of \eqref{eq:p} is a convex combination, and thus we obtain 
		\begin{align*}
			\frac{\sum_{i=0}^{j}\arw(\bar{T}_i)}{j+1}\geq \sum_{i=0}^j \frac{\lambda_i d_i^c({\cC})}{(i+1)\Delta}\geq \min_{0\leq i\leq j}\Big\{\frac{d_i^c(\cC)}{(i+1)\Delta}\Big\}.
		\end{align*}
	\end{proof}
	
	\begin{lemma}[The rate of $\tilde{\cC}_{\phi}$]\label{le:dim}
		The $\ff_q[D]$-rank\footnote{This $\ff_q[D]$-rank is often informally referred to as the $\ff_q$-dimension of convolutional codes over $\ff_q$.} of $\tilde{\cC}$ is at least $(k-(m+1)(1-r)\Delta)n$ and the rate of $\tilde{\cC}_{\phi}$ satisfies 
		$$
		\frac{\log_{q^{r\Delta}}|\tilde{\cC}_{\phi}|}{n}\geq \frac{k}{r\Delta}-(m+1)\Big(\frac{1}{r}-1\Big).
		$$
	\end{lemma}
	\begin{proof}
		We begin by examining the $\ff_q$-dimension of $\cB$, which is equal to the $\ff_q[D]$-rank of $\tilde{\cC}$. Note that $\cB=\cB^{(1)}\cap\cB^{(2)}$.
		Moreover, from \eqref{eq:bc1} and \eqref{eq:bc2}, we have $\dim_{\ff_q}\cB^{(1)}=n\dim_{\ff_q}\cB_1=nk$ and $\dim_{\ff_q}\cB^{(2)}=n\dim_{\ff_q}\cB_2^{m+1}=n(m+1)r\Delta$. Therefore, the $\ff_q$-dimension of $\cB$ satisfies 
		\begin{align*}
			\dim_{\ff_q} \cB
			&\geq \dim_{\ff_q}\cB^{(1)} + \dim_{\ff_q}\cB^{(2)} - n(m+1)\Delta\\
			&=(k-(m+1)(1-r)\Delta)n.
		\end{align*}
		Since $\phi$ is bijective, it follows that
		\begin{align*}
			\frac{\log_{q^{r\Delta}}|\tilde{\cC}_{\phi}|}{n} =\frac{\log_{q^{r\Delta}}|\tilde{\cC}|}{n}= \frac{{\rank_{\ff_q[D]}}\tilde{\cC}}{nr\Delta} = \frac{\dim_{\ff_q}\cB}{nr\Delta} \geq \frac{k}{r\Delta}-(m+1)\Big(\frac{1}{r}-1\Big).
		\end{align*}
	\end{proof}
	
	In general, it is difficult to estimate the degree of $\tilde{\cC}_{\phi}$. However, if the generator matrix $G(D)$ for the convolutional code $\cC$ is reduced and has equal row degrees, then we have the following result.
	
	\begin{lemma}[The degree of $\tilde{\cC}_{\phi}$]\label{le:degree}
		Assume the $k\times \Delta$ generator matrix $G(D)$ of $\cC$ is reduced and has equal row degrees, i.e., $\nu_i(G)={\delta(\cC)}/{k}$ for all $i\in\pz{k}$. Then the degree of $\tilde{\cC}_{\phi}$ satisfies
		$$\frac{\delta(\tilde{\cC}_{\phi})}{n}\leq \min\Big\{\frac{\delta(\cC)}{r\Delta},\frac{(m+1)\delta(\cC)}{k}\Big\}.$$
	\end{lemma}
	
	\begin{proof}
		By construction of the generator matrix $\tilde{G}(D)$, we have $\nu_i(\tilde{G})\leq {\delta(\cC)}/{k}$ for all $i\in\pz{\tilde{k}}$. It follows that the degree of $\tilde{\cC}$ is at most $\tilde{k}\delta(\cC)/k$. At the same time, by construction, we also have $\tilde{k}\leq \min\{\dim_{\ff_q}\cB^{(1)},\dim_{\ff_q}\cB^{(2)}\}=\min\{nk,n(m+1)r\Delta\}$. Therefore, the degree of $\tilde{\cC}$ is at most $\min\{n\delta(\cC),n(m+1)r\Delta\delta(\cC)/k\}$. 
		Since $\phi$ is bijective, the minimum number of states in a generator digraph for $\tilde{\cC}_{\phi}$ is the same as that of $\tilde{\cC}$. Hence, the degree of $\tilde{\cC}_{\phi}$ satisfies 
		\[
		\delta(\tilde{\cC}_{\phi})\leq \min\Big\{\frac{n\delta(\cC)}{r\Delta},\frac{n(m+1)\delta(\cC)}{k}\Big\}.
		\]
	\end{proof}
	
	Combing Lemma~\ref{le:cd}, \ref{le:dim} and \ref{le:degree}, and taking $\cC$ be an MDP convolutional code, we obtain the following theorem.
	\begin{theorem}\label{thm:main}
		Let $\cG_0$ be a $(\Delta,\gamma)$ spectral expander and let $\cC$ be an $(\Delta,k)$ MDP convolutional code over $\ff_q$. Assume that the $k\times \Delta$ generator matrix $G(D)$ of $\cC$ is reduced and has equal row degrees.\footnote{This assumption is not restrictive. Indeed, many known constructions of MDP convolutional codes are based on a reduced generator matrix with equal row degrees.} If $r\to 1$ and $\gamma\to 0$, then $\tilde{\cC}_{\phi}$ is an $(n,M)$ trellis code over $\ff_{q^{r\Delta}}$ with
		\begin{align*}
			&\frac{\delta(\tilde{\cC}_\phi)}{n}\leq \frac{\delta(\cC)}{\Delta},\\
			&\frac{\log_{q^{r\Delta}}M}{n}\to\frac{k}{\Delta},\\
			&\frac{d_j^c(\tilde{\cC}_{\phi})}{(j+1)n}\geq 1-\frac{k}{\Delta} \text{ for all $j\in\nz{L(\cC)}$},\\
			&{\frac{\df(\tilde{\cC}_{\phi})}{(j+1)n}\geq 1-\frac{k}{\Delta} \text{ for all $j\in\nz{L(\cC)}$}.}
		\end{align*}
	\end{theorem}
	\begin{proof}
		As $r\to 1$, by Lemma~\ref{le:degree}, we have ${\delta}(\tilde{\cC}_{\phi})/n\leq \delta(\cC)/\Delta$.
		Moreover, if $r\to 1$, then by Lemma~\ref{le:dim}, we have $(\log_{q^{r\Delta}}M)/n\geq k/\Delta$.
		At the same time, by construction, the rate of $\tilde{\cC}$ is no larger than that of $\cC$, i.e., $k/\Delta$. Since $\phi$ is bijective, the same holds for $\tilde{\cC}_{\phi}$, namely, $(\log_{q^{r\Delta}}M)/n\leq k/\Delta$. Hence, we obtain 
		\begin{align*}
			\frac{\log_{q^{r\Delta}}M}{n}\to\frac{k}{\Delta}.
		\end{align*}
		By Lemma~\ref{le:cd}, as $\gamma\to 0$, the $j$th column distance of $\tilde{\cC}_{\phi}$ satisfies
		\begin{align*}
			\frac{d_j^c(\tilde{\cC}_{\phi})}{(j+1)n}\geq \min_{0\leq i\leq j}\Big\{\frac{d_i^c(\cC)}{(i+1)\Delta}\Big\}.
		\end{align*}
		Since $\cC$ is an MDP convolutional code, for all $j\in\nz{L(\cC)}$ we further obtain
		\begin{align}
			\frac{d_j^c(\tilde{\cC}_{\phi})}{(j+1)n}&\geq\frac{d_j^c(\cC)}{(j+1)\Delta}=\frac{(j+1)(1-k/\Delta)+1/\Delta}{j+1}\geq 1-\frac{k}{\Delta}.\nonumber
		\end{align}
		{Lastly, since $\tilde{\cC}_{\phi}$ is $\ff_q$-linear, there exists $C^*(D)=\sum_{i=0}^{\infty}C_i^*D^i\in\tilde{\cC}_{\phi}$ with $C^*\neq 0$ such that $\wt(C^*(D))=\df(\tilde{\cC}_{\phi})$. It follows that for all $j\geq 0$,
			\begin{align*}
				d_j^c(\tilde{\cC}_{\phi})\leq \wt(C_0^*,\ldots,C_j^*)\leq \df(\tilde{\cC}_{\phi}).
			\end{align*} Hence, for all $j\in\nz{L(\cC)}$ we obtain
			\begin{align*}
				\frac{\df(\tilde{\cC}_{\phi})}{(j+1)n}\geq\frac{d_j^c(\tilde{\cC}_{\phi})}{(j+1)n}\geq 1-\frac{k}{\Delta}.
		\end{align*}}
	\end{proof}
	\begin{remark}
		On the one hand, by the Singleton bound \eqref{eq:tsb} for trellis codes, the free distance of the code $\tilde{\cC}_{\phi}$ in Theorem~\ref{thm:main} satisfies
		\begin{align*}
			\frac{\df(\tilde{\cC}_{\phi})}{n}\leq \Big(1-\frac{\log_{q^{r\Delta}} M}{n}\Big)\left( \left\lfloor\frac{\delta(\tilde{\cC}_{\phi})}{\log_{q^{r\Delta}} M}\right\rfloor+\frac{\delta(\tilde{\cC}_{\phi})}{n-\log_{q^{r\Delta}} M}+1\right)+\frac{1}{n}.
		\end{align*}
		Combining this with Theorem~\ref{thm:main} (and letting $r\to 1$) we obtain 
		\begin{align*}
			\frac{\df(\tilde{\cC}_{\phi})}{n}\leq \Big(1-\frac{k}{\Delta}\Big)\left( \frac{\delta({\cC})}{k}+\frac{\delta({\cC})}{\Delta-k}+1\right)+o(1).
		\end{align*}
		On the other hand, Theorem~\ref{thm:main} (letting $\gamma\to 0$) also gives
		\begin{align*}
			\frac{\df(\tilde{\cC}_{\phi})}{n} \geq (L(\cC)+1)\Big(1-\frac{k}{\Delta}\Big), 
		\end{align*} where $L(\cC)$ is the maximum profile length of the $(\Delta,k)$ MDP convolutional code $\cC$. 
		If $(\Delta-k)\mid \delta(\cC)$, i.e., $\cC$ is strongly-MDS, then the above inequality becomes
		\begin{align*}
			\frac{\df(\tilde{\cC}_{\phi})}{n}
			&\geq \Big(\frac{\delta(\cC)}{k}+\frac{\delta(\cC)}{\Delta-k}+1\Big)\Big(1-\frac{k}{\Delta}\Big).
		\end{align*}
		Hence, in this case, the free distance of $\tilde{\cC}_{\phi}$ relative to $n$ is within an $o(1)$ term of the maximum possible allowed by the Singleton bound \eqref{eq:tsb} for trellis codes.
	\end{remark}

	\section{Concluding remarks}\label{sec:re}
	Since expander-based constructions often admit low-complexity encoding and decoding, it is natural to expect the same for the trellis codes developed in this paper. An important direction for future work is to make this intuition precise by designing efficient encoding/decoding algorithms for our construction. In addition, while our results yield trellis codes over constant-size alphabets with rate-distance trade-offs close to those of MDP convolutional codes, it remains open whether one can construct nearly-MDP convolutional codes over constant-size finite fields.

	\appendices
	\addtocontents{toc}{\protect\setcounter{tocdepth}{0}}
	
	\section{Proof of Claim~\ref{cl:trellis-B2}}
	For each $j\geq 0$, we have $c_j=\sum_{i=0}^mx_{j-i}\tilde{G}_i$ where $x_{j-i}\in\ff_q^{\tilde{k}}$ with $x_{j-i}=0$ if $j-i<0$. Therefore, $(c_j)_{\cE(u_{j,s})}=(x_j\tilde{G}_0)_{\cE(u_{j,s})}+(x_{j-1}\tilde{G}_1)_{\cE(u_{j,s})}+\cdots+(x_{j-m}\tilde{G}_m)_{\cE(u_{j,s})}$. Since the columns of $\tilde{G}_i$ are indexed by $E_i$ and the ordering $E_j$ is consistent with $E_0$ for $j\geq 1$, one may relabel the columns of $\tilde{G}_i$ by mapping $(u_{i,s},v_{i,t})$ to $(u_{j,s},v_{j,t})$ while preserving the ordering of the columns, where $s,t\in\pz{n}$. Thus, one may write $(x_{j-i}\tilde{G}_i)_{\cE(u_{j,s})}=(x_{j-i}\tilde{G}_i)_{\cE(u_{i,s})}$. By Proposition~\ref{prop:B} and \eqref{eq:def:B}, $(x_{j-i}\tilde{G}_i)_{\cE(u_{i,s})}$ is a codeword in $\cB_2$ for all $i\in\nz{m}$. It follows that $(c_j)_{\cE(u_{j,s})}=\sum_{i=0}^m(x_{j-i}\tilde{G}_i)_{\cE(u_{i,s})}\in\cB_2$.
	
	\section{Proof of Claim~\ref{cl:trellis-C}}
	For $j\geq 0$, we have 
	\begin{align}
		(c_0,\ldots,c_j)
		=x_0(\tilde{G}_0, \ldots, \tilde{G}_j)
		+x_1(0,			  \tilde{G}_{0}, \ldots,      \tilde{G}_{j-1})
		+\cdots
		+x_j(0,			     \ldots, 0,\tilde{G}_0),
	\end{align} where $x_0,\ldots,x_j\in\ff_q^{\tilde{k}}$ and $\tilde{G}_j=0$ if $j>m$.  Since the orderings on $E_i$ and $E_j$ are consistent, one may relabel the columns of $\tilde{G}_i$ by $E_j$ for any $i\in\nz{m}$ and $j\geq 0$, while preserving the ordering of the columns. Therefore, for $i\in\nz{m}$ one may write $((x_i\tilde{G}_0)_{\cE(v_{i,s})},\ldots,(x_i\tilde{G}_{j-i})_{\cE(v_{j,s})})=((x_i\tilde{G}_0)_{\cE(v_{0,s})},\ldots,(x_i\tilde{G}_{j-i})_{\cE(v_{j-i,s})})$. Note that by Proposition~\ref{prop:B} and \eqref{eq:def:B}, $((x_i\tilde{G}_0)_{\cE(v_{0,s})},\ldots,(x_i\tilde{G}_{j-i})_{\cE(v_{j-i,s})})$ can be obtained by either truncating the last $(m-j+i)\Delta$ coordinates of some codeword in $\cB_1$ if $j-i\leq m$ or appending $(j-i-m)\Delta$ zeros to the end of some codeword in $\cB_1$ if $j-i>m$. It then follows that 
	\begin{align*}
		((c_0)_{\cE(v_{0,s})},\ldots,(c_j)_{\cE(v_{j,s})})
		&=((x_0\tilde{G}_0)_{\cE(v_{0,s})},\ldots,(x_0\tilde{G}_{j})_{\cE(v_{j,s})})
		+(0,(x_1\tilde{G}_0)_{\cE(v_{1,s})},\ldots,(x_{1}\tilde{G}_{j-1})_{\cE(v_{j,s})})+\cdots\\
		&\qquad\qquad
		+(0,\ldots,0,(x_j\tilde{G}_0)_{\cE(v_{j,s})})\\
		&=((x_0\tilde{G}_0)_{\cE(v_{0,s})},\ldots,(x_0\tilde{G}_{j})_{\cE(v_{j,s})})
		+(0,(x_1\tilde{G}_0)_{\cE(v_{0,s})},\ldots,(x_{1}\tilde{G}_{j-1})_{\cE(v_{j-1,s})})+\cdots\\
		&\qquad\qquad
		+(0,\ldots,0,(x_j\tilde{G}_0)_{\cE(v_{0,s})})
	\end{align*} 
	is a codeword of the $(\Delta,k)$ convolutional code $\cC$ truncated at time instant $j$.

	\section{Proof of Lemma~\ref{le:rw}}
	If $S_j$ is empty, then $T_j$ is also empty and we have $\arw(\bar{T}_j)=0$. In this case, the inequality is trivial, so let us assume $S_j$ is nonempty.
	Since $S_j$ is nonempty (which implies $T_j$ is also nonempty), by a version of the expander mixing lemma \cite[Proposition~3.3]{roth2006improved}, we have the following upper bound on $|Y_j|$:
	\begin{align*}
		|Y_j|\leq \Big( (1-\gamma)\frac{|S_j||T_j|}{n} + \gamma\sqrt{{|S_j|}{|T_j|}}\Big)\Delta.
	\end{align*}
	Since $|Y_j|\geq |S_j|\cdot\theta\Delta$ and $|Y_j|=|T_j|\cdot\arw(\bar{T}_j)\Delta$, it follows that
	\begin{align}
		|S_j|\theta \leq |T_j|\arw(\bar{T}_j)  \leq\Big( (1-\gamma)\frac{|S_j||T_j|}{n} + \gamma\sqrt{{|S_j|}{|T_j|}}\Big).\label{eq:set-edge}
	\end{align}
	From the second inequality in \eqref{eq:set-edge}, we have
	\begin{align}
		\arw(\bar{T}_j) \leq (1-\gamma)\frac{|S_j|}{n} + \gamma\sqrt{\frac{|S_j|}{|T_j|}}.\label{eq:chain}
	\end{align}
	From the first inequality in \eqref{eq:set-edge}, we have ${|S_j|}/{|T_j|}\leq \arw(\bar{T}_j)/\theta $ and thus it follows from \eqref{eq:chain} that
	\begin{align*}
		\arw(\bar{T}_j) \leq (1-\gamma)\frac{|S_j|}{n} + \gamma\sqrt{\frac{\arw(\bar{T}_j)}{\theta}}\leq (1-\gamma)\frac{|S_j|}{n} + \gamma\sqrt{{\theta^{-1}}}.
	\end{align*}
	Therefore, rearranging the above inequality, we obtain
	\begin{align*}
		\frac{|S_j|}{n}\geq \frac{\arw(\bar{T}_j)-\gamma\sqrt{\theta^{-1}}}{1-\gamma}.
	\end{align*}
	
	\section{Proof of Lemma~\ref{le:partition}}
	
	Note that
	\begin{align}
		{\sum_{i=0}^{j}\arw(\bar{T}_i)\Delta}
		&=\sum_{i=0}^{j}\sum_{s\in T_i}\frac{\wt((c_i)_{\cE(v_{i,s})})}{|T_i|}\nonumber\\
		&\geq \sum_{i=0}^j{\sum_{s\in T_0\cup \dots\cup T_i}\frac{\wt((c_i)_{\cE(v_{i,s})})}{|T_0\cup \dots\cup T_i|}},\label{eq:union}
	\end{align} where \eqref{eq:union} follows since $(c_i)_{\cE(v_{i,s})}\neq 0$ for $s\in T_i$ and $(c_i)_{\cE(v_{i,s})}= 0$ for $s\in (T_0\cup\dots\cup T_i)\setminus T_i$.
	
	For $\ell\in\nz{j}$, let $a_{\ell}=|T_0\cup \cdots \cup T_{\ell}|$ and define $a_{-1}=0$. Since $a_0=|T_0|>0$, we have $0=a_{-1}<a_0\leq a_1\leq \cdots\leq a_{j}$. Moreover, we have $a_{\ell+1}-a_{\ell} = |T_{\ell+1}\setminus(T_0\cup\cdots\cup T_{\ell})|$ for $\ell\in\nz{j-1}$. Let $b_{\ell}=a_{\ell}^{-1}$ for $\ell\in\nz{j}$ and define $b_{j+1}=0$. Then we have
	\begin{align}
		\sum_{i=0}^j{\sum_{s\in T_0\cup \dots\cup T_i}\frac{\wt((c_i)_{\cE(v_{i,s})})}{|T_0\cup \dots\cup T_i|}}
		&=\sum_{\ell=0}^{j}\sum_{s\in T_\ell\setminus (T_0\cup \dots\cup T_{\ell-1})} \sum_{i=\ell}^j
		{b_i}{\wt((c_i)_{\cE(v_{i,s})})},\label{eq:sum}
	\end{align}
	where $T_0\cup \dots\cup T_{\ell-1}:=\emptyset$ if $\ell = 0$. 
	The above identity follows from noticing that one can partition the set $T_0\cup\dots\cup T_j$ into subsets $T_0, T_1\setminus T_0, \ldots, T_j\setminus(T_0\cup\dots\cup T_{j-1})$ and sum over each subset.\footnote{We adopt the convention that a sum over the empty set is zero.}
	For $\ell\in\nz{j}$, it holds that 
	\begin{align}
		\sum_{s\in T_\ell\setminus (T_0\cup \dots\cup T_{\ell-1})} \sum_{i=\ell}^j
		{b_i}{\wt((c_i)_{\cE(v_{i,s})})}
		&=
		\sum_{s\in T_\ell\setminus (T_0\cup \dots\cup T_{\ell-1})} 
		\Big(
		(b_{\ell}-b_{\ell+1})\wt((c_{\ell})_{\cE(v_{\ell,s})}) \nonumber\\
		& \qquad + (b_{\ell+1}-b_{\ell+2})( \wt((c_{\ell})_{\cE(v_{\ell,s})})+  \wt((c_{\ell+1})_{\cE(v_{\ell+1,s})})) + \cdots\nonumber\\
		&\qquad +(b_j-b_{j+1})\sum_{i=0}^{j-\ell}\wt((c_{\ell+i})_{\cE(v_{\ell+i,s})})
		\Big)\nonumber\\
		&\geq 
		\sum_{s\in T_\ell\setminus (T_0\cup \dots\cup T_{\ell-1})}
		\sum_{i=0}^{j-\ell}(b_{\ell+i}-b_{\ell+i+1})d_i^c({\cC})\label{eq:cdb}\\
		&= 
		|T_\ell\setminus (T_0\cup \dots\cup T_{\ell-1})|\sum_{i=0}^{j-\ell}(b_{\ell+i}-b_{\ell+i+1})d_i^c({\cC})\nonumber\\
		&=(a_{\ell}-a_{\ell-1})\sum_{i=0}^{j-\ell}(b_{\ell+i}-b_{\ell+i+1})d_i^c({\cC}).\nonumber
	\end{align} To justify \eqref{eq:cdb}, we note that for each $s\in T_\ell\setminus (T_0\cup \dots\cup T_{\ell-1})$ we have $\wt((c_{\ell})_{\cE(v_{\ell,s})})>0$ and $\wt((c_{\ell'})_{\cE(v_{\ell',s})})=0$ for all $\ell'\in \nz{\ell-1}$ since $s\notin T_{\ell'}$. Therefore, by Claim~\ref{cl:trellis-C}, the sequence $((c_{\ell})_{\cE(v_{\ell,s})},\ldots,(c_{\ell+i})_{\cE(v_{\ell+i,s})})$ is a nonzero truncated codeword of $\cC$ that first becomes nonzero at time instant $\ell$ for all $i\in\nz{j-\ell}$. It then follows that $\sum_{t=\ell}^{\ell+i}\wt((c_{t})_{\cE(v_{t,s})})\geq d_{i}^c(\cC)$ for all $i\in\nz{j-\ell}$, which yields \eqref{eq:cdb}.
	Hence, we obtain
	\begin{align}
		\sum_{\ell=0}^{j}\sum_{s\in T_\ell\setminus (T_0\cup \dots\cup T_{\ell-1})} \sum_{i=\ell}^j
		{b_i}{\wt((c_i)_{\cE(v_{i,s})})}
		&\geq \sum_{\ell=0}^{j} (a_{\ell}-a_{\ell-1})\sum_{i=0}^{j-\ell}(b_{\ell+i}-b_{\ell+i+1})d_i^c({\cC})\nonumber\\
		&= \sum_{i=0}^{j} \sum_{\ell=0}^{j-i} (a_{\ell}-a_{\ell-1}) (b_{\ell+i}-b_{\ell+i+1})d_i^c({\cC}).\label{eq:convex}
	\end{align}
	For $i\in\nz{j}$, define
	\begin{align*}
		\lambda_i
		&=\frac{i+1}{j+1}\sum_{\ell=0}^{j-i}(a_{\ell}-a_{\ell-1})(b_{\ell+i}-b_{\ell+i+1}).
	\end{align*}
	Clearly, $\lambda_i\geq 0$. Moreover, 
	\begin{align*}
		\sum_{i=0}^j\lambda_i &= \frac{1}{j+1}\sum_{i=0}^j(i+1)\sum_{\ell=0}^{j-i}(a_{\ell}-a_{\ell-1})(b_{\ell+i}-b_{\ell+i+1})\\
		&=\frac{1}{j+1}\sum_{\ell=0}^{j}(a_{\ell}-a_{\ell-1})\sum_{i=0}^{j-\ell}(i+1)(b_{\ell+i}-b_{\ell+i+1})\\
		&=\frac{1}{j+1}\sum_{\ell=0}^{j}(a_{\ell}-a_{\ell-1})\sum_{i=0}^{j-\ell}b_{\ell+i}\\
		&=\frac{1}{j+1}\sum_{\ell=0}^{j}(a_{\ell}-a_{\ell-1})\sum_{i=\ell}^{j}b_{i}\\
		&=\frac{1}{j+1}\sum_{i=0}^{j}b_{i}\sum_{\ell=0}^{i}(a_{\ell}-a_{\ell-1})\\
		&=\frac{1}{j+1}\sum_{i=0}^{j}b_{i}a_{i}\\
		&=1.
	\end{align*}
	Hence, it follows from \eqref{eq:union}, \eqref{eq:sum} and \eqref{eq:convex} that 
	\begin{align*}
		{\sum_{i=0}^{j}\arw(\bar{T}_i)\Delta}
		&\geq (j+1)\sum_{i=0}^j \frac{\lambda_i d_i^c({\cC})}{i+1},
	\end{align*} where $\lambda_i\geq 0$ for all $i\in\nz{j}$ and $\sum_{i=0}^j\lambda_i=1$.

	\bibliographystyle{IEEEtran}
	\bibliography{GoodConvolutionalCodes}
	
\end{document}